%% file: main.tex
\documentclass{article}
\input{admin/preamble2}

\input{admin/macros}
\usepackage{graphicx} % Required for inserting images
\usepackage{url}
\usepackage{float}

\usepackage{authblk}

\usepackage[utf8]{inputenc}
\usepackage{graphicx}
\usepackage{parskip}
\usepackage{gensymb}
\restylefloat{table}

\usepackage[dvipsnames]{xcolor}
\restylefloat{table}
\usepackage{caption}
\usepackage{subcaption}
\usepackage{array}
\usepackage{multirow}

\usepackage{nameref}
\usepackage{bbm}

\usepackage{mathtools}
\usepackage{dsfont}
\usepackage{braket}
\usepackage{amsfonts}
\usepackage{amssymb}
\usepackage{CJKutf8}
\usepackage{multicol}
\numberwithin{equation}{section}

\title{Time evolution of impurity models and their universality for quantum computation}
\author[1, 2]{N. C. Mai Pham}
\author[1]{Raul A. Santos}
\date{\today}
\affil[1]{Phasecraft Ltd.}
\affil[2]{University College London}

\begin{document}

\maketitle

\begin{abstract}

Impurity Hamiltonians are systems of $N$ fermionic modes where $\mathcal{O}(1)$ of them interact among themselves via quartic (or higher order) fermion terms, while coupling quadratically with $\mathcal{O}(N)$ bath modes. Without the quartic interactions, these systems are classically simulable with $\mathcal{O}(N^3)$ resources. In \cite{brod2012geometriesmatchgates} it was proved that the time-dependent evolution of these systems can perform universal quantum computation. The question of whether or not this remains true for time-independent evolution remains open. Here we prove that the time evolution of generic time-independent impurity Hamiltonians on $\mathcal{O}(N)$ qubits is universal on $N$ qubits if the input state is a product state of fermions in any single particle basis. In our proof we find that for a computation of depth $S$, the size of the impurity scales as $\mathcal{O}(S\log S)$. 

%Quantum impurity model, is used as the simpler problem to study correlated material via Dynamical Mean Field Theoy (DMFT). While the time evolution of such time-dependent Hamiltonian is proven to be able to perform universal computation, it’s hard to realise its physical implementation. In this work, given a desired calculation and the time-dependent Hamiltonian that achieve it, we have a procedure to construct a time-independent Hamiltonian on a larger space. We prove that by preparing the initial state to be a product state in a particle basis state, using just a change of basis, allowing the system to evolve with the previously defined time-independent Hamiltonian and considering a subspace of the final system, we will achieve the same universal computational power as that of the time-dependent one. 
    
\end{abstract}

%\tableofcontents

\section{Introduction}

It has been known for some time that the dynamics of a system can be used to perform computational tasks, and even simple systems can perform (reversible) universal \textit{classical} computation \cite{fredkin1982conservative}. 

It is believed that quantum mechanics provides a different computational model \cite{bernstein1993quantum} and as such the connection between time dynamics and quantum computation has also been explored. In particular, it has been shown that multi-particle quantum walks or scattering of momentum states in the Fermi-Hubbard model can perform universal quantum computation \cite{Childs_2013_QuantumWalk, Bao_2015}, a result that can be considered the quantum version of \cite{fredkin1982conservative}. In these systems, particles move in a lattice and can experience interactions on all sites.

On the other hand, the time evolution of restricted systems like free-fermions over $N$ modes can only perform universal quantum computation over $\log(N)$ modes, signalling the crucial role of interactions. A natural question follows: How many interactions are actually needed for the dynamics of a system to perform universal quantum computation? In \cite{brod_childs2013computational}, a partial answer to this question was given, within the framework of the circuit model. There, the authors analysed the power of circuits based on matchgates in geometries not a cycle or not strictly linear, and proved that these circuits are indeed universal for quantum computation. Doing a Jordan-Wigner transformation \cite{jordanwigner1993} to map qubits to fermionic modes, the circuit model studied in \cite{brod_childs2013computational} can be understood as the evolution under a time-dependent Hamiltonian of an \textit{impurity} model, i.e. a fermionic model with interactions (four fermion terms) acting on $\mathcal{O}(1)$ modes.

Impurity models have a rich history. First formulated in \cite{anderson1961localized}, they were used to understand the anomalous resistance behaviour with temperature of some metals, which ultimately led to the discovery of the Kondo effect \cite{kondo1964resistance}. Nowadays, they are used to connect calculations of model systems with real materials through dynamical mean field theory (DMFT) \cite{kotliar2006electronic_dmft}, where the dynamical response of the impurity model is used self-consistently to approximate the Green's function of a material. \cite{Bravyi_2017_Impurity} poses the question of whether the time evolution of a \textit{time-independent} system with $\mathcal{O}(1)$ interactions can perform universal computation. Our work proves that universal computation can be achieved with time-independent systems given the number of impurities scales as $S\ln{(S)}$, where $S$ is the computation depth. This has renewed interest in the computational power of impurity models, subsequently, \cite{arunachalam2026impurity} claims universality of impurity models that are time-independent and fixed in impurity size.

Our result is shown in the following theorem:

% In this work we give a partial answer to this question, in the form of the following theorem

% \rs{\begin{theorem}[Time evolution of impurity models and universal quantum computation]\label{thm:1}
% The evolution with the Hamiltonian
% \begin{equation} 
%     H_{\rm ind} \coloneq \jindex H_j + \sum_{m,r=-M}^{M} \left (\gm \cN{r}{m} + {\gm}^* \cNdag{r}{m} \right ) \left ( 1 - {2} \sum_{p=-M}^{M} \cNN{p}{p}\right ) 
% \end{equation}
% with
% \begin{equation}
%     H_j \coloneq \sum_{r=-M}^{M} r \cjcj + \sum_{m,r=-M}^{M} \left ( \fmj \cj{r}{m} + {\fmj}^* \cjdag{r}{m} \right )
% \end{equation}
% acting on $N(2M+1)$ fermionic modes $c_{j,r}$ is universal on $\Omega(N)$ qubits. For a computation of depth $S$, the size of the impurity $M= \mathcal{O}(S\ln{S})$.
% \end{theorem}}

\begin{theorem}[Time evolution of impurity models and universal quantum computation]\label{thm:1}
The evolution with the Impurity Hamiltonian
\begin{equation} \label{eq:H}
    H_{\rm ind} \coloneq \sum_{j=1}^{N-2} B_j + \sum_{m,r=-M}^M \left (\tilde{g}^{(m)} \cN{r}{m} + \tilde{g}^{(m)*} \cNdag{r}{m} \right ) \left ( 1 - {2} \sum_{p=-M}^M c^\dagger_{N-1,p}c_{N-1,p}\right ),
\end{equation}
with
\begin{equation}
    \tilde{B}_j\coloneq   \sum_{m,r=-M}^M \left ( \tilde{f}_j^{(m)} \cj{r}{m} + \tilde{f}_j^{(m)*} \cjdag{r}{m} \right ).
\end{equation}
moving with constant velocity in a magnetic field that varies linearly in space acting on $N(2M+1)$ fermionic modes $c_{j,r}$ is universal on $\Omega(N)$ qubits. For a computation of depth $S$, the size of the impurity $M= \mathcal{O}(S\ln{S})$.
\end{theorem}

We prove that the time-independent impurity model Hamiltonian $H_{\rm ind}$ above is universal. In particular, given any computation $\mathcal{C}$, we provide a construction of a time-independent impurity model that approximates $\mathcal{C}$ with arbitrarily small error. Our proof technique requires an impurity size that grows with the computation depth $S$ as $\mathcal{O}(S\ln{S})$. After the first version of this work appeared, Arunachalam \emph{et al.}~\cite{arunachalam2026impurity} settled this question, showing that simulating the time evolution of a time-independent impurity Hamiltonian with a fixed, constant-size impurity (a quartic interaction on $24$ fermionic modes) is BQP-complete; in their construction the circuit is encoded in an initial Slater determinant of wave packets that propagate freely through the bath and scatter off the impurity. Our construction is complementary: the circuit is encoded in the couplings of the Hamiltonian rather than in the initial state, at the cost of an impurity of size $\mathcal{O}(S\ln S)$.

In our construction, we are also able to produce an explicit mapping between a given quantum circuit and a time-dependent Hamiltonian, given in \cref{lem:first}.

\textbf{The proof idea}

The proof of \cref{thm:1} follows from a chain of reductions. To prove that Hamiltonian simulation of a time-independent impurity Hamiltonian is universal, we have to show that given a calculation, the time evolution of such model can approximate it to arbitrary accuracy. As a starting point we map the universality result proven in \cite{brod_childs2013computational} to the evolution generated by a time-dependent impurity Hamiltonian, which we call $H_{\rm BC}(t)$. We define the periodic and differentiable extension of $H_{\rm BC}(t)$ with period equal to the computation time $P$. This Hamiltonian corresponds to the one obtained by truncating the Fourier series defining the coefficients of $H_{\rm BC}(t)$ . We call this Hamiltonian $H_{\rm tr}^{[M]}(t)$. Finally we show that in a magnetic field that varies linearly in space, a carefully constructed time-independent Hamiltonian $H_{\rm ind}$ on a larger number of modes generates the same time-dependent evolution as the truncated Hamiltonian $H_{\rm tr}^{[M]}(t)$, once we restrict it to a subspace.

%We can subsequently define the Fourier coefficients of the periodic and differentiable extension of these functions, which are used to construct a time-independent impurity Hamiltonian $H_{\rm ind}$ on a larger Hilbert space. We show that this $H_{\rm ind}$, in the interaction picture and in the momentum basis, generates the same dynamics as the original $H_{\rm BC}$ on a certain subspace. 

Taking into account the errors that arise in each step, we bound the total error between the evolution generated by $H_{\rm ind}$ and $H_{\rm BC}(t)$. Finally, we calculate how much bigger the Hilbert space on which $H_{\rm ind}$ acts must be in order for the total error to be arbitrarily small. In the end, we find that the overhead must be of the order $\mathcal{O}(S\ln{S})$ for a circuit depth $S$. These steps are shown schematically in \cref{fig:flowchart}.

\begin{figure} 
    \centering
    \includegraphics[width=\linewidth]{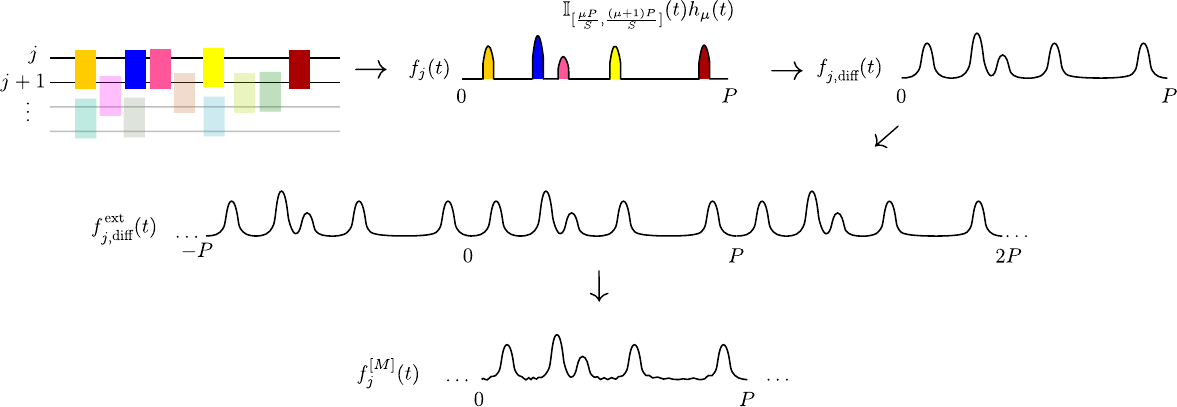}
    \caption{Main ideas of the proof. The circuit that implements universal quantum computation is mapped into a time-dependent evolution over time $P$, with coefficient functions $f_j$  composed of sums of local functions $\mathbbm{1}_{[\frac{\mu P}{S}, \frac{(\mu + 1) P}{S}]}(t)h_\mu (t)$. By discarding the indicator functions we can construct a differentiable extension of $f_j$ that we denote $f_{j,{\rm diff}}$. These coefficient functions can be extended to be periodic over a period $P$. These extensions are denoted as $f_{j,\rm diff}^{\rm ext}$. Finally we truncate the Fourier series representation of these periodic functions. This truncation is denoted as $f_j^{[M]}$. The evolution under a time dependent Hamiltonian with coefficient functions $f_j^{[M]}$ can be implemented through a time-independent Hamiltonian $H_{\rm ind}$ with $\mathcal{O}(M)$ ancillas.}
    % The purple, orange, green curves are the coefficient functions corresponding to the gates $O_1, O_2, O_3$, respectively. These functions can be truncated (see 2. $H_{\rm BC}(t)$), or extended to be differentiable and periodic (see 3. $H_{diff, per}(t))$. The lighter curves (see 4.) are Fourier decompositions of the differentiable and periodic extension of the original coefficient functions. The lighter lines (see 5.) are the Fourier coefficients, which are used as the constant coefficient functions for the time-independent $H_{\rm ind}$. In this figure, the Fourier series is truncated to only $3$ terms (see 4.). }
    \label{fig:flowchart}
\end{figure}

The first reduction is discussed in \cref{subsec:red1}. Its periodic extension is discussed in \cref{subsec:extension}. We show the construction of the time-independent Hamiltonian in \cref{subsec:construction_H_indep}. \cref{subsec:final} shows that $H_{\rm ind}$ approximates the time-dependent $H_{\rm BC}(t)$. By adding accumulatively the errors occurred from each calculation step, we explicitly calculate the error bound between the time evolution of these two Hamiltonians and show that the number of modes grows as $S\ln{S}$ with the number of operations $S$, we can achieve arbitrarily small error. In \cref{sec:conclusion} we discuss some implications of this result.

\section{Proof of main theorem}

\subsection{From the circuit model to a time-dependent Hamiltonian}
\label{subsec:red1}

Given a quantum circuit, one can decompose its gates into single-qubit rotation gates and Heisenberg interaction gates (such as $R_{XX}(\theta)$ gates which performs a unitary transformation of the form $\e{-i\frac{\theta}{2}(X\otimes X)}$ ). The CNOT gate, with which the single-qubit rotation gates can also form a universal gates set, can be written with the Heisenberg interaction gates via $\text{CNOT} = \e{-i\frac{\pi}{4}}R_Y^1\left(\frac{-\pi}{2}\right)R_X^1\left(\frac{-\pi}{2}\right)R_X^2\left(\frac{-\pi}{2}\right)R_{XX}^{1, 2}\left(\frac{\pi}{2}\right)R_Y^2\left(\frac{\pi}{2}\right)$, where $R_X^1(\theta)$ is the $X$-rotation gate with angle $\theta$ on qubit $1$, and so on.
% \rs{Mai, can you provide a definition of Heisenberg interaction gates and show that the model CNOT + single qubit gates can be mapped to this?. This is trivial, but good to be pedagogical.}

Hence, WLOG, we can assume that a quantum circuit of $S$ gates can be written as a series of instructions
$\mathcal{C} = \{G_\mu(\theta_\mu)\}_{\mu=0}^{S-1}$ 
where $G_\mu$ is a quantum gate, either a rotation gate or a Heisenberg interaction gate, and $\theta_\mu$ is its parameter ($0\leq \theta_\mu\leq 2\pi$), $\mu$ indexes the order in which these gates are applied. Using the exponential form of such gates,
\begin{equation}\label{eq:circ-H1}
    \mathcal{C} = \{G_\mu(\theta_\mu)\}_{\mu=0}^{S-1} \quad \text{corresponds to} \quad \prod_{\mu = 0}^{S-1} \e{-i\frac{\theta_\mu}{2}{O_\mu}}
\end{equation}
where $O_\mu$ is the corresponding operator to the quantum gate $G_\mu$. For example, a circuit $\{R_X^1(\theta_1),R_Z^1(\theta_2),R_{XX}^{1,2}(\theta_3)\}$ corresponds to $\e{-i\frac{\theta_1}{2}X_1}\e{-i\frac{\theta_2}{2}Z_2}\e{-i\frac{\theta_3}{2}X_1X_2}$. With this in mind, we prove the following lemma
\begin{lemma}\label{lem:first}
Any circuit $\mathcal{C}$ described as in \cref{eq:circ-H1} corresponds to the time evolution with a time-dependent Hamiltonian $H(t) = \sum_{\mu=0}^{S-1} \mathbbm{1}_{[\frac{\mu P}{S}, \frac{(\mu + 1) P}{S}]}(t)h_\mu (t) O_\mu$ over some total evolution time $P$. Here
    $\mathbbm{1}_{[a, b]}(t) = \begin{cases}
    1 \quad t \in [a, b] \\
    0 \quad \text{otherwise}
\end{cases}$
is the indicator function. The functions $h_\mu(t)$ satisfy $\int_{-\infty}^\infty \mathbbm{1}_{[\frac{\mu P}{S}, \frac{(\mu + 1) P}{S}]}h_\mu(t)dt=\frac{\theta_\mu}{2}$, but are otherwise arbitrary. 
\end{lemma}

\begin{proof}
We first revisit properties of the unitary time evolution operator $U(t_i, t_f)$ generated by a given time-dependent  Hamiltonian $H(t)$
$$U(t_i, t_f) = \mathcal{T}\e{-i\int_{t_i}^{t_f}H(s) ds}$$
where the symbol $\mathcal{T}$ represents the time ordering of operators.
The time evolution operator satisfies the group property
\begin{align}\label{eq:group_prop}
    U(t_0, t_2) &= U(t_0, t_1)U(t_1, t_2) \\
    \mathcal{T}\e{-i\int_{t_0}^{t_2}H(s) ds} &= \mathcal{T}\e{-i\int_{t_0}^{t_1}H(s) ds}\mathcal{T}\e{-i\int_{t_1}^{t_2}H(s) ds}
\end{align}
Given the Hamiltonian
%consisting of $S$ operator terms with piecewise continuous coefficient functions on regular intervals of computational time $P$ as $H(t) = \sum_{\mu=0}^{S-1} \mathbbm{1}_{[\frac{\mu P}{S}, \frac{(\mu + 1) P}{S}]}(t)h_\mu (t) O_\mu$
\begin{equation}
     H(t) = \sum_{\mu=0}^{S-1} \mathbbm{1}_{[\frac{\mu P}{S}, \frac{(\mu + 1) P}{S}]}(t)h_\mu (t) O_\mu
 \end{equation}
%where for all $0 \leq \mu \leq S-1$, $\mathbbm{1}_{[\frac{\mu P}{S}, \frac{(\mu + 1) P}{S}]}(t) = \begin{cases}
%    1 \quad t \in [\frac{\mu P}{S}, \frac{(\mu + 1) P}{S}] \\
%    0 \quad \text{otherwise}
%\end{cases}$, $O_\mu$ is an operator and 
%\begin{equation} \label{eq:para-area}
%    \int_{\frac{\mu P}{S}}^{\frac{(\mu + 1) P}{S}} h_\mu(t) \, dt =: \frac{\theta_\mu}{2}\ \quad \text{for} \ 0 \leq \mu \leq S-1
%\end{equation}
The corresponding time evolution generated by $H(t)$ is then
\begin{align}
    \mathcal{T}\e{-i\int_{0}^{P}H(s) ds} &= \prod_{\mu=0}^{S-1} \mathcal{T}\e{-i\int_{\frac{\mu P}{S}}^{\frac{(\mu + 1) P}{S}}H(s) ds} 
    % = \prod_{j=0}^S \mathcal{T}\e{-i\int_{a_j}^{a_{j+1}}O_j ds} 
    % = \prod_{j=0}^S \e{-i(a_{j+1}-a_j)O_j} \\
    = \prod_{\mu=0}^{S-1} \e{-i\int_{\frac{\mu P}{S}}^{\frac{(\mu + 1) P}{S}} h_\mu(s) O_\mu ds} 
    = \prod_{\mu=0}^{S-1} \e{-i\frac{\theta_\mu}{2}O_\mu}
\end{align}
Here, the first equality follows from the property \cref{eq:group_prop} and the second is a consequence of the definition of $H(s)$, where within an interval the time dependence is only through the function $h$. The last equality follows from the normalisation of the functions $h_\mu$. The last expression corresponds to a circuit in \cref{eq:circ-H1}, which consists of rotation gates or Heisenberg interaction gates implemented in series, whose angles are $\theta_\mu$. 
\end{proof}

Studying the computational power of matchgates, Brod and Childs \cite{brod_childs2013computational} have shown that matchgates defined in any connected graph that is neither a line nor a cycle can perform universal quantum computation. In particular, considering a line with a loop at the end (see \cref{fig:cylinder} a), they discuss how to generate the universal gate set of the entangling gate CZ and single-qubit gates. They use two qubits adjacent to the qubit corresponding to the 3-degree vertex as ancillas and encode one logical qubit by two physical qubits. They can perform a CZ gates on two logical qubits neighboring the ancillas by applying 11 f-SWAP gates (which are matchgates) on the $4$ physical qubits and the $2$ ancillas. Lastly, they argue that any logical qubit can be moved to any desired location using only f-SWAP gates, concluding that gates of this type on nearest neighboring qubits on such a graph can perform universal computations.
% \rs{Mai please add the specific gate set that they generate in their paper please}.
Using a Jordan-Wigner transformation \cite{jordanwigner1993}, a Hamiltonian consisting of matchgates between qubits on the edges of \cref{fig:cylinder}a, can be mapped into a fermionic Hamiltonian corresponding to an impurity model of the form
    \begin{equation} \label{eq:H_BC}
        H_{\rm BC}(t) \coloneq \jindex f_j^{\rm BC}(t) (c_j^\dag c_{j+1} + c_{j+1}^\dag c_j) + g^{\rm BC}(t) (1 - 2c_{N-1}^\dag c_{N-1}) (c_{N-2}^\dag c_N + c_N^\dag c_{N-2}).
    \end{equation}
where $t\in[0,P]$ with $P$ the total computation time. Using the indicator function defined above, we can write this Hamiltonian as
    \begin{equation} \label{eq:H2ways}
        H_{\rm BC}(t) \coloneq \sum_{\mu=0}^{S-1} \mathbbm{1}_{[\frac{\mu P}{S}, \frac{(\mu + 1) P}{S}]}(t)h_\mu(t) O_\mu 
    \end{equation} 
The two presentations of $H_{\rm BC}$ \cref{eq:H_BC} and \cref{eq:H2ways} are equivalent under the definitions
\begin{align}
        f_j^{\rm BC}(t)\coloneq &\sum_{\mu: O_\mu = c_j^\dag c_{j+1} + c_{j+1}^\dag c_j} \mathbbm{1}_{[\frac{\mu P}{S}, \frac{(\mu + 1) P}{S}]}(t)h_\mu(t)  \\
        g^{\rm BC}(t)\coloneq &\sum_{\mu: O_\mu =  (1 - 2c_{N-1}^\dag c_{N-1}) (c_{N-2}^\dag c_N + c_N^\dag c_{N-2})} \mathbbm{1}_{[\frac{\mu P}{S}, \frac{(\mu + 1) P}{S}]}(t)h_\mu(t) 
\end{align}

Note that for a given circuit the depth $S$ is fixed. We call the functions $f_j^{\rm BC}(t)$ for $j=1\dots N-2$ and $g^{\rm BC}(t)$ the coefficient functions of $H_{\rm BC}(t)$.

\subsection{Periodic differentiable extension and finite Fourier expansion} \label{subsec:extension}

From \cref{eq:H2ways}, it is clear that the coefficient functions of the Hamiltonian $H_{\rm BC}(t)$ are defined in the interval $t\in[0,P]$. We also see that the coefficient functions might be discontinuous, hence not differentiable, at the end points of each interval, due to the indicator functions $\mathbbm{1}_{[\frac{\mu P}{S}, \frac{(\mu + 1) P}{S}]}(t)$. 
For differentiable functions $h_\mu$, a natural way to construct a differentiable and periodic counterpart of the coefficient functions of $H_{\rm BC}(t)$ is to take the coefficient functions of the new Hamiltonian to be $h_\mu^{\rm ext}(t) \coloneq \sum_{l \in \mathbb{Z}} h_\mu(t+lP)$
\begin{align}
H_{\rm diff}(t) \coloneq \sum_{\mu=0}^{S-1} h_\mu^{\rm ext}(t)\, O_\mu =\sum_{\mu=0}^{S-1}\sum_{l \in \mathbb{Z}} h_\mu(t + lP) \, O_\mu,
\end{align}
which determine the periodic and differentiable extensions
\begin{align}
        f_{j,\rm diff}^{\rm ext}(t)\coloneq & \sum_{l\in \mathbb{Z}}f_{j,\rm diff}(t+lP)=\sum_{\substack{l\in\mathbb{Z}\\\mu: O_\mu = c_j^\dag c_{j+1} + c_{j+1}^\dag c_j}} h_\mu(t+ lP)  \\
        g_{\rm diff}^{\rm ext}(t)\coloneq &\sum_{m\in \mathbb{Z}}g_{\rm diff}(t+lP)=\sum_{\substack{l\in \mathbb{Z} \\\mu: O_\mu =  (1 - 2c_{N-1}^\dag c_{N-1}) (c_{N-2}^\dag c_N + c_N^\dag c_{N-2})}} h_\mu(t+mP) 
\end{align}

%We can extend their definition over the whole real line by introducing the periodic extensions
% \begin{align}
%     f_j^{\rm ext}(t+mP)\coloneq f_j^{\rm BC}(t),\\
%     g^{\rm ext}(t+mP)\coloneq g^{\rm BC}(t),
% \end{align}
% \begin{align}
%     f_j^{\rm ext}(t)\coloneq \sum_{l\in \mathbb{Z}}f_j^{\rm BC}(t+lP),\\
%     g^{\rm ext}(t)\coloneq \sum_{l\in \mathbb{Z}} g^{\rm BC}(t+lP),
% \end{align}
%with $m\in\mathbb{Z}$. 
We see that to make the coefficient functions of the newly defined $H_{\rm diff}(t)$ differentiable, we keep those of $H_{\rm BC}$, but without the indicator functions $\mathbbm{1}_{[\frac{\mu P}{S}, \frac{(\mu + 1) P}{S}]}(t)$. Note that this construction comes with a caveat that $h_\mu$ must decay exponentially fast outside the interval $[0, P]$ for the infinite sums above to converge. This set of extended functions is periodic with period $P$. By the assumptions of \cref{lem:first}, the functions $\{f_j^{\rm ext}\}_{j=1}^{N-2}, g^{\rm ext}$ satisfy the Dirichlet sufficiency conditions \cite{Cornelius_book} in the interval $[0,P]$. This implies that each function can be approximated pointwise by a Fourier series of the form
\begin{align}
    f_{j,\rm diff}^{\rm ext}(t)=\sum_{m=-\infty}^{\infty}e^{i\frac{2\pi }{P}mt}f_{j}^{(m)},\quad g^{\rm ext}_{\rm diff}(t)=\sum_{m=-\infty}^{\infty}e^{i\frac{2\pi }{P}mt}g^{(m)},
\end{align}
where $\fmj$ and $\gm$ are the Fourier coefficients of the functions $f_j^{\rm ext}(t)$ and $g^{\rm ext}(t)$ respectively,
i.e.
\begin{equation} \label{eq:Fourier_coeff}
    \fmj \coloneq \frac{1}{P} \int_0^P \e{i\frac{2\pi}{P}mt} f_{j, \rm diff}^{\rm ext}(t) dt, \quad \gm \coloneq \frac{1}{P} \int_0^P \e{i\frac{2\pi}{P}mt} g^{\rm ext}_{\rm diff}(t) dt.
\end{equation}
The reality of the functions $g^{\rm ext}$ imposes the relation $(g^{(m)})^*=g^{(-m)}$ between its Fourier coefficients, and similarly for all $f_j$.

We define the (time-dependent) Hamiltonian
\begin{equation} \label{eq:H_trunc}
    H_{\rm tr}^{[M]}(t) \coloneq \jindex f_j^{[M]}(t) (c_j^\dag c_{j+1} + c_{j+1}^\dag c_j) + g^{[M]}(t) (1 - 2c_{N-1}^\dag c_{N-1}) (c_{N-2}^\dag c_N + c_N^\dag c_{N-2}),
\end{equation}
 with coefficient functions that are the truncated versions of $f_j^{\rm ext}$ and $g^{\rm ext}$, i.e.
 \begin{align}\label{eq:F_coeff}
        f_j^{[M]}(t)\coloneq \sum_{m=-M}^{M}e^{i\frac{2\pi }{P}mt}f_j^{(m)},\quad g^{[M]}(t)\coloneq \sum_{m=-M}^{M}e^{i\frac{2\pi }{P}mt}g^{(m)}.
 \end{align}
The chain of these extensions is shown in \cref{fig:flowchart}. In the next subsection we explain how a specific time-independent Hamiltonian generates the same evolution as \cref{eq:H_trunc} in a subspace.

\subsection{The time-independent Hamiltonian}
\label{subsec:construction_H_indep}

\begin{figure}
    \centering
    \includegraphics[width=1\linewidth]{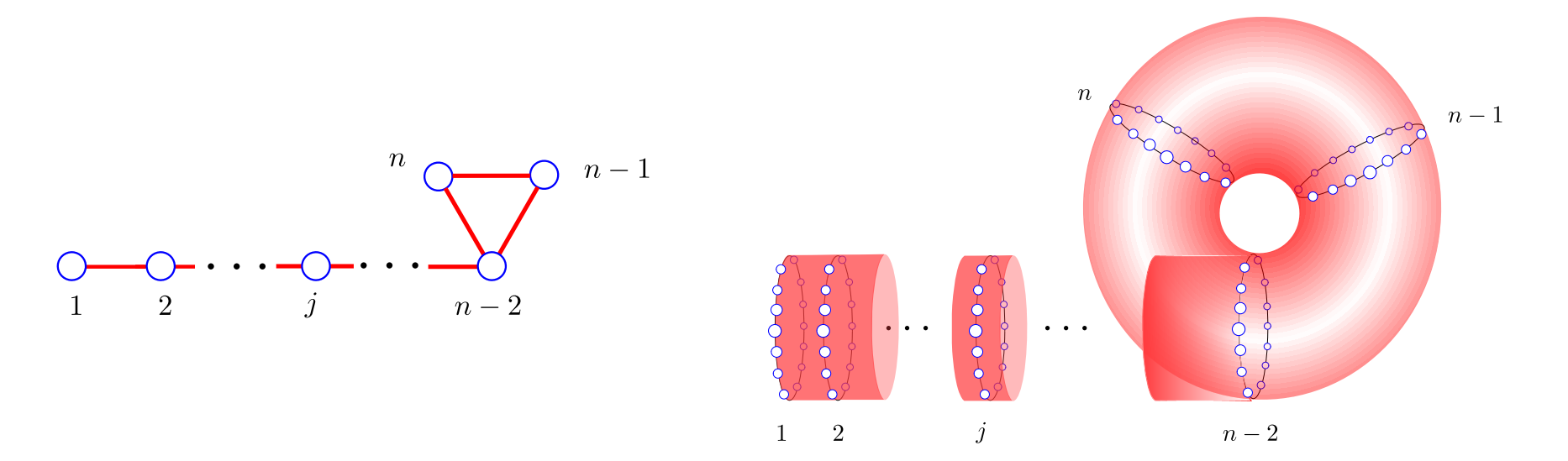}
    \caption{Left: Original interaction graph of the Hamiltonian $H_{\rm BC}$ (\ref{eq:H_BC}). The blue circles represent lattice sites, while the red lines represent interactions in terms of Paulis. Right: The graph representing the time-independent Hamiltonian $H_{\rm ind}$ (\ref{eq:H}) where each site in the fermionic picture becomes a ring of $2M+1$ sites, interacting with their neighbours along the cylinder in a translational-invariant way across the circumference of the cylinder.}
    \label{fig:cylinder}
\end{figure}

% The overall Hamiltonian in the lab frame:
% \begin{align}
%     H_{\rm frame} &\coloneq \sum_{i=1}^Q \frac{p_i^2}{2m_q} +  \sum_{j=1}^{N-2} B_j \nonumber \\
%     &+ \sum_{m,r=-M}^M \left (g^{(m)} e^{-i\frac{q B_z(\textbf{x}) R^2 \pi m}{2M+1}}  \cN{r}{m} + g^{(m)*}e^{i\frac{q B_z(\textbf{x}) R^2 \pi m}{2M+1}} \cNdag{r}{m} \right ) \left ( 1 - {2} \sum_{p=-M}^M c^\dagger_{N-1,p}c_{N-1,p}\right ), 
% \end{align} 
% with
% \begin{equation}
%     B_j\coloneq   \sum_{m,r=-M}^M \left ( f_j^{(m)} e^{-i\frac{q B_z(\textbf{x}) R^2 \pi m}{2M+1}} \cj{r}{m} + f_j^{(m)*} e^{i\frac{q B_z(\textbf{x}) R^2 \pi m}{2M+1}} \cjdag{r}{m} \right ).
% \end{equation}

% where $q$ is the charge of the particles, $Q$ is the number of particles. $p_i$ are the momentum of each particles, and $x$ is the position of lattice sites. Note that $p_i$ and $x$ are classical variables.

We can now define the impurity Hamiltonian on the cylinder of \cref{fig:cylinder}. In the absence of any external field, it is the time-independent Hamiltonian
\begin{equation} \label{eq:H}
    H_{\rm ind} \coloneq \sum_{j=1}^{N-2} B_j + \sum_{m,r=-M}^M \left (\tilde{g}^{(m)} \cN{r}{m} + \tilde{g}^{(m)*} \cNdag{r}{m} \right ) \left ( 1 - {2} \sum_{p=-M}^M c^\dagger_{N-1,p}c_{N-1,p}\right ),
\end{equation}
with
\begin{equation}
    \tilde{B}_j\coloneq   \sum_{m,r=-M}^M \left ( \tilde{f}_j^{(m)} \cj{r}{m} + \tilde{f}_j^{(m)*} \cjdag{r}{m} \right ).
\end{equation}
We assume periodic boundary conditions in the second index of the fermion operators i.e. $c_{j,r+2M+1} = c_{j, r}$ (see \cref{fig:cylinder}). The couplings $\tilde{f}_j^{(m)}$ and $\tilde{g}^{(m)}$ differ from the Fourier coefficients $\fmj$ and $\gm$ defined in \cref{eq:Fourier_coeff} only by fixed phases, given in \cref{eq:matching} below.

We now discuss the set-up in which the time-independent Hamiltonian $H_{\rm ind}$ generates the same time evolution as $H_{\rm tr}^{[M]}(t)$.

\subsubsection{Set-up: an impurity model moving through a linear magnetic field}
\label{subsec:moving}

To generate the time dependence of $H_{\rm tr}^{[M]}(t)$, we let the lattice of $H_{\rm ind}$ move through a static magnetic field that increases linearly in space (see \cref{fig:field}). Such a magnetic field can be generated by anti-Helmholtz coils or solenoids with varying turns density. The modes of $H_{\rm ind}$ sit on the surface of a cylinder of radius $R$ whose axis is the $z$ axis, with mode $(j,r)$ located at
\begin{equation}\label{eq:positions}
    \bm{x}_{j,r} = \left(R\cos\varphi_r,\, R\sin\varphi_r,\, z_j\right),\qquad \varphi_r = \frac{2\pi r}{2M+1},\qquad z_j = j\ell .
\end{equation}
The first index labels the position along the axis, with layer spacing $\ell$, and the second the angular position on a ring, so the periodic boundary condition $c_{j,r+2M+1}=c_{j,r}$ is the statement that $\varphi_{r+2M+1}=\varphi_r$. In the lab frame there is no electric field, and the magnetic field is
\begin{equation}\label{eq:B_field}
    \bm{B}(\bm{x}) = \lambda\left(-\frac{x}{2},\, -\frac{y}{2},\, z\right) = -\frac{\lambda\rho}{2}\,\hat{\bm{\rho}} + \lambda z\,\hat{\bm{z}},
\end{equation}
where $\rho$ is the distance from the axis. Its axial component increases linearly along the axis. The radial component is required by $\nabla\cdot\bm{B}=0$, and since $\nabla\times\bm{B}=0$ the field can be produced by external currents, as near the centre of a pair of anti-Helmholtz coils. Being axially symmetric, the field acts in the same way on every site of a ring. The lattice moves along its axis with constant velocity $v\ll c$. We describe the fermions in the rest frame of the lattice, where the sites are fixed at $\bm{x}_{j,r}$ and the field at the point $\bm{x}$ is $\bm{B}(\bm{x}+vt\,\hat{\bm{z}})$. In this frame the axial field at layer $j$, and with it the magnetic flux $\pi R^2\lambda(z_j+vt)$ through every ring, grows linearly in time, while the radial component stays constant. We take the fermions to be confined to the surface of the cylinder and choose the vector potential
\begin{equation}\label{eq:A_field}
    \bm{A}(\bm{x}) = \frac{\lambda z}{2}\,(-y,\, x,\, 0) = \frac{\lambda z\rho}{2}\,\hat{\bm{\varphi}},\qquad \nabla\times\bm{A}=\bm{B}.
\end{equation}
In the reference frame of the lattice the potentials are $\bm{A}'(\bm{x},t) = \bm{A}(\bm{x}+vt\,\hat{\bm{z}})$ and $\phi'(\bm{x},t) = -\bm{v}\cdot\bm{A}(\bm{x}+vt\,\hat{\bm{z}}) = 0$, where the scalar potential vanishes because $\bm{A}$ has no $z$ component. The whole effect of the field on the lattice is therefore a time-dependent vector potential, which enters the hopping terms through the Peierls substitution.

\begin{figure}[t]
    \centering
    \includegraphics[width=\linewidth]{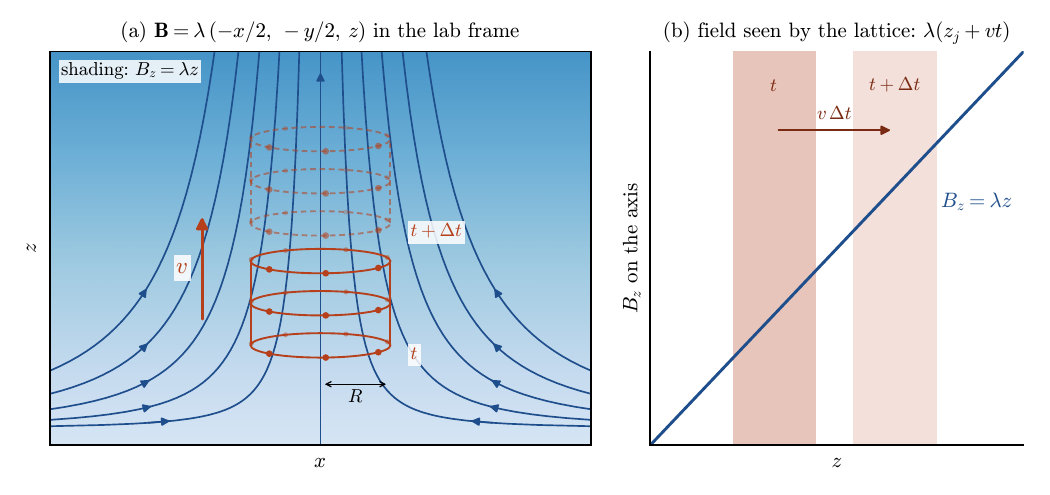}
    \caption{Set-up of \cref{subsec:moving}. (a) Field lines in the $x$--$z$ plane of the static lab-frame field $\bm{B}$ of \cref{eq:B_field}; the shading shows the axial component $B_z=\lambda z$, which increases linearly along the axis. The modes of $H_{\rm ind}$ sit on rings of $2M+1$ sites (drawn schematically) on a cylinder of radius $R$, which moves along the axis with constant velocity $v$. The dashed copy shows the lattice a time $\Delta t$ later. (b) Axial field along the axis. In the rest frame of the lattice the axial field at layer $j$ is $\lambda(z_j+vt)$, so the flux through every ring grows linearly in time.}
    \label{fig:field}
\end{figure}

\subsubsection{Peierls substitution}
\label{subsec:peierls}

In a tight-binding model, a magnetic field enters through the Peierls substitution \cite{peierls1933theorie, luttinger1951effect}: each hopping term acquires a phase given by the line integral of the vector potential along the hopping path,
\begin{equation}\label{eq:peierls}
    c_a^\dagger c_b \;\longrightarrow\; \exp\left(iq\int_{\bm{x}_b}^{\bm{x}_a}\bm{A}\cdot d\bm{l}\right) c_a^\dagger c_b ,
\end{equation}
where $q$ is the charge of the fermions (with $\hbar=1$) and $c_a^\dagger c_b$ moves a particle from site $b$ to site $a$. The prescription extends to time-dependent vector potentials in a gauge with vanishing scalar potential \cite{graf1995electromagnetic}, which is the situation here. Number operators such as $c_{N-1,p}^\dagger c_{N-1,p}$ are unaffected. We take the path to be the geodesic on the cylinder along the direction of the hopping (which depends on the sign of $m$), i.e. the helix joining $\bm{x}_b$ and $\bm{x}_a$. For the hopping term $c_{j,r}^\dagger c_{j+1,r+m}$ with $m\in\{-M,\dots,M\}$, it turns through the angle $\Delta\varphi_m = 2\pi m/(2M+1)$ . On the cylinder, $\bm{A}'\cdot d\bm{l} = \frac{\lambda R^2}{2}(z+vt)\,d\varphi$, so by parametrising the helix as $z(s)=z_j+s\ell$, $\varphi(s)=\varphi_r+s\Delta\varphi_m$ with $s\in[0,1]$, we find
\begin{equation}\label{eq:phase}
    q\int_{\bm{x}_{j,r}}^{\bm{x}_{j+1,r+m}}\bm{A}'\cdot d\bm{l} = \frac{q\lambda R^2}{2}\,\Delta\varphi_m\int_0^1 \left(z_j+s\ell+vt\right) ds = \kappa\, m\left[\left(j+\tfrac{1}{2}\right)\ell + vt\right],
\end{equation}
with $\kappa\coloneq \pi q\lambda R^2/(2M+1)$. The term $c_{j,r}^\dagger c_{j+1,r+m}$ moves a particle along this helix in the opposite direction, so it acquires the phase $e^{-i\kappa m[(j+1/2)\ell+vt]}$, and its Hermitian conjugate the opposite phase. The same computation for $c_{N-2,r}^\dagger c_{N,r+m}$, whose helix spans two layers with midpoint at height $(N-1)\ell$, gives $e^{-i\kappa m[(N-1)\ell+vt]}$. The time-dependent part of every phase depends on $m$ but not on $r$, so it is the same for all hopping terms. Choosing the velocity and the couplings in \cref{eq:H} as
\begin{equation}\label{eq:matching}
    \kappa v = \frac{2\pi}{P},\qquad \tilde{f}_j^{(m)} = e^{i\kappa m\ell\left(j+\frac{1}{2}\right)}\fmj,\qquad \tilde{g}^{(m)} = e^{i\kappa m\ell(N-1)}\gm,
\end{equation}
the Hamiltonian of the lattice in its rest frame is
\begin{align}\label{eq:Hbar}
    \bar{H}(t) &= \jindex \sum_{m,r=-M}^M \left( \fmj \e{-i\frac{2\pi}{P}mt} \cj{r}{m} + {\fmj}^* \e{i\frac{2\pi}{P}mt} \cjdag{r}{m}\right) \nonumber \\
    &+ \sum_{m,r=-M}^M \left( \gm \e{-i\frac{2\pi}{P}mt} \cN{r}{m} + {\gm}^* \e{i\frac{2\pi}{P}mt} \cNdag{r}{m}\right) \left( 1 - 2\sum_{p=-M}^M\cNN{p}{p} \right).
\end{align}
If the layer spacing satisfies $\kappa\ell\in 4\pi\mathbb{Z}$, the phases in \cref{eq:matching} are trivial and $\tilde{f}_j^{(m)}=\fmj$, $\tilde{g}^{(m)}=\gm$. The evolution of the lattice over one period is $U_P\coloneq\mathcal{T}\exp\left(-i\int_0^P\bar{H}(t)dt\right)$.
Performing a Fourier rotation in the second index (denoted by $FT_2$) we have
\begin{align}
FT_2 c_{j,r} FT_2^\dagger = \frac{1}{\sqrt{2M+1}}\sum_{k=-M}^M e^{i\frac{2\pi}{2M+1}kr}c_{j,k},
\end{align}
such that a direct computation leads to
\begin{align}\label{eq:FT_Hbar}
    &FT_2 \bar{H}(t) FT_2^\dag 
    = \sum_{j=1}^{N-2}\sum_{m,k=-M}^{M}\left(f_{j}^{(m)}e^{-i\frac{2\pi}{P}m(t-\frac{Pk}{2M+1})}c_{j,k}^{\dagger}c_{j+1,k}+f_{j}^{*(m)}e^{i\frac{2\pi}{P}m(t-\frac{Pk}{2M+1})}c_{j+1,k}^{\dagger}c_{j,k}\right)\\
    &+ \sum_{m,k=-M}^M \left (  \gm \e{-i\frac{2\pi}{P}m (t - \frac{kP}{2M+1})} c_{N-2,k}^\dag c_{N, k} +  {\gm}^* \e{i\frac{2\pi}{P}m (t - \frac{Pk}{2M+1})} c_{N,k}^\dag c_{N-2, k} \right ) \left( 1 - 2 \sum_{p=-M}^Mc_{N-1,p}^\dag c_{N-1, p} \right).  \nonumber
\end{align}

% \begin{figure}[htp]
%     \centering
%     \includegraphics[width=\linewidth]{Images/graph1.png}
%     \caption{How $h_\mu^{\rm BC}(t), h_\mu^{\rm ext}(t)$ and $h_\mu^{[M]}(t)$ are related (and similarly, $f_j^{\rm BC}(t), f_j^{\rm ext}(t)$ and $f_j^{[M]}(t)$; and $g^{\rm BC}(t), g^{\rm ext}(t)$ and $g^{[M]}(t)$). Here, the parameter function $h_\mu^{\rm BC}(t)$ from the original Hamiltonian $H_{\rm BC}$ is taken to be a truncated Gaussian function (see 2.). We see that given $h_\mu^{BC}(t)$, $h_\mu^{\rm ext}(t)$ is the periodic smooth extension of $h_\mu^{BC}(t)$ (see 3.), and $h_\mu^{[M]}(t)$ is the truncation of $h_\mu^{\rm ext}(t)$ (see 4.).}
%     % \rs{I think this is a good idea, although the last step should show the truncation of the Fourier series.}
%     \label{fig:flowchart}
% \end{figure}

 The unitary $\mathcal{U}_P\coloneq  FT_2U_PFT_2^\dag$ then satisfies
\begin{align}\nonumber
     \mathcal{U}_P= \mathcal{T}e^{-i\int_0^PFT_2\bar{H}(s)FT_2^\dagger ds}= \mathcal{T}e^{-i\int_0^P\mathcal{H}(s)ds}.
\end{align}
The Hamiltonian $\mathcal{H}(t)\coloneq FT_2\bar{H}(t)FT_2^\dagger$ has at least $2M+1$ conserved quantities, one for each transversal momentum sector as $[\mathcal{H}(t), N_k] = 0, \forall k \in \{-M,..., M\}, \text{ where } N_k \coloneq \sum_{j=1}^N c_{j,k}^\dag c_{j,k}$. This means that the evolution generated by $\mathcal{H}(t)$ on a state $\ket{\Psi_0}$ with initial occupations only in the $k = 0$ sector $(N_k = 0 \text{ if } k \neq 0)$ will remain in that subspace.
This implies that on the state $|\Psi_0\rangle$ the Hamiltonian $\mathcal{H}(t)$ only acts nontrivially on the modes $c_{l,0}$, i.e $\mathcal{H}(t)|\Psi_0\rangle =\mathcal{H}_0(t)|\Psi_0\rangle$, where
\begin{align}
    \mathcal{H}_0(t)&\coloneq \sum_{j=1}^{N-2}\sum_{m=-M}^{M}\left(f_{j}^{(m)}e^{-i\frac{2\pi}{P}mt}c_{j,0}^{\dagger}c_{j+1,0}+f_{j}^{*(m)}e^{i\frac{2\pi}{P}mt}c_{j+1,0}^{\dagger}c_{j,0}\right)\\
    &+ \sum_{m=-M}^M \left (\gm \e{-i\frac{2\pi}{P}mt} c_{N-2,0}^\dag c_{N, 0} +  {\gm}^* \e{i\frac{2\pi}{P}m t} c_{N,0}^\dag c_{N-2, 0} \right ) \left( 1 - 2c_{N-1,0}^\dag c_{N-1, 0} \right).  \nonumber
\end{align}
performing the sums and comparing with \cref{eq:F_coeff}, we have
\begin{align}
    \mathcal{H}_0(t)=\jindex f_j^{[M]}(t) (c_{j,0}^\dag c_{j+1,0} + c_{j+1,0}^\dag c_{j,0}) + g^{[M]}(t) (1 - 2c_{N-1,0}^\dag c_{N-1,0}) (c_{N-2,0}^\dag c_{N,0} + c_{N,0}^\dag c_{N-2,0}).
\end{align}
Note that $\mathcal{H}_0$ is equivalent to $H_{\rm tr}^{[M]}$ (see \cref{eq:H_trunc}) under the substitution $c_{j}\rightarrow c_{j, 0}$. This implies the connection
\begin{align}
    \mathcal{U}_P\ket{\Psi_0} =\mathcal{T}e^{-i\int_0^P{H}_{\rm tr}^{[M]}(s)ds} \ket{\Psi_0}.
\end{align}
We have thus shown that the impurity model $H_{\rm ind}$, moving through the field of \cref{eq:B_field} for a time $P$, generates in its rest frame the same time evolution as the time-dependent Hamiltonian $H_{\rm tr}^{[M]}(t)$ up to a change of basis. By construction $H_{\rm tr}^{[M]}(t)$ is an approximation of $H_{\rm BC}(t)$. As $H_{\rm BC}(t)$ generates universal quantum computation, we just have to prove that the approximation error between the unitary generated by $H_{\rm tr}^{[M]}(t)$ can be made arbitrary small. The rest of the proof quantifies the error between the evolution with $H_{\rm tr}^M(t)$ and $H_{\rm BC}(t)$.

\subsection{$H_{\rm tr}^{[M]}$ as the approximation for $H_{\rm BC}$ with bounded error} \label{subsec:final}

The error between the evolutions generated by $H_{\rm tr}^M(t)$ and $H_{\rm BC}(t)$ can be upper bounded using the integral representation of the error. For two unitaries $U_1(t)=\mathcal{T}\e{-i\int_0^tH_1(s) ds}$ and $U_2(t)=\mathcal{T}\e{-i\int_0^tH_2(s) ds}$ generated by the Hamiltonians $H_1(t)$, $H_2(t)$ respectively, the error function $E \coloneq U_1^\dag(t)U_2(t) - 1$ satisfies the first order differential equation
\begin{align}
    \frac{d}{dt}E(t)=i\left(U_1^\dag (t)(H_1(t) - H_{2}(t))U_{2}(t)\right), 
\end{align}
integrating and using the initial condition $E(0)=0$ this is equivalent to
\begin{align}
     E(t) = i\int_0^t U_1^\dag(t) (H_1(t) - H_{2}(t))U_{2}(t) dt.
\end{align}
Applying the operator norm and the triangle inequality leads to
\begin{align}
    \|E(t)\| &\leq \int_0^t \|U^\dag(t) (H(t) - H_{apx}(t))U_{apx}(t)\| dt \leq \int_0^t \|H(t) - H_{apx}(t) \| dt \label{eq:error}
\end{align}
where the invariance of the operator norm under unitary transformations was used in the last equality.
This means that we can bound the error between evolutions by bounding the difference between their generators. A direct computation of the error between $H_{\rm BC}(t)$ and $H_{\rm tr}^M(t)$ is subtle as the coefficient functions of $H_{\rm BC}(t)$ are piecewise continuous
and approximation with a finite Fourier series can suffer from the Gibbs phenomenon \cite{Vretblad2003-jt}. To avoid this we first approximate the coefficient functions of $H_{\rm BC}(t)$ by continuous functions and then approximate those by the coefficient functions of $H_{\rm tr}^{[M]}(t)$. Denoting the intermediate Hamiltonian by $H_{\rm diff}(t)$ and using the triangle inequality together with \cref{eq:error}
\begin{align}\label{ineq:error}
    E \coloneq \| \mathcal{T}e^{-i\int_0^PH_{\rm BC}(s)ds}-\mathcal{T}e^{-i\int_0^PH_{\rm tr}^{[M]}(s)ds}\|\leq \int_0^P\left(\| H_{\rm BC}(s)-H_{\rm diff}(s)\|+\| H_{\rm diff}(s)-H_{\rm tr}^{[M]}(s)\|\right)ds
\end{align}

We denote the total error by $E$, the difference between $H_{\rm BC}(t)$ and $H_{\rm diff}(t)$ by $E_{\rm Area}$, and the difference between $H_{\rm diff}(t)$ and $H_{\rm tr}^{[M]}(t)$ by $E_{\rm Fourier}$ to have $E = E_{\rm Area} + E_{\rm Fourier}$.

So far, we have not discussed the specific form of the coefficient functions of $H_{\rm BC}$. In this section, we take a closer look to explicitly calculate the error between the evolution generated by $H_{\rm tr}^{[M]}(t)$  and $H_{\rm BC}(t)$ over time $P$. 

% \rs{Mai please take it from here    }

% Firstly, we want to give the explicit construction of the intermediate Hamiltonian $H_{\rm diff}(t)$. Recall that $H_{\rm BC}(t)$ can be written as in \cref{eq:H2ways} (see \cref{fig:flowchart}.2):
% \begin{equation} \tag{\ref{eq:H2ways}}
%     H_{\rm BC}(t) \coloneq \sum_{\mu=0}^{S-1} \mathbbm{1}_{[\frac{\mu P}{S}, \frac{(\mu + 1) P}{S}]}(t)h_\mu(t) O_\mu 
% \end{equation}
% A natural way to construct a differentiable counterpart of the coefficient functions of $H_{\rm BC}(t)$ is to take the coefficient functions of the new Hamiltonian to be $h_\mu^{\rm ext}(t) \coloneq \sum_{l \in \mathbb{Z}} h_\mu(t+lP)$
% \begin{align}
% H_{\rm diff}(t) \coloneq \sum_{\mu=0}^{S-1} h_\mu^{\rm ext}(t)\, O_\mu =\sum_{\mu=0}^{S-1}\sum_{l \in \mathbb{Z}} h_\mu(t + lP) \, O_\mu,
% \end{align}
% which is simply a superposition of $h_\mu(t)$'s, instead of a concatenation (see \cref{fig:flowchart}.3). 

We define the intermediate Hamiltonian $H_{\rm diff}$ to be $P$-periodic as well. This is motivated by the use of Fourier truncation of the coefficient functions in \cref{eq:H}. We use the freedom to choose functions $h_\mu(t)$ that minimize the error between the $H_{\rm BC}(t)$ and $H_{\rm diff}(t)$. A good candidate are Gaussian functions with mean around the middle of the interval where the indicator functions are non-zero, specifically
\begin{equation} \label{eq:hmu}
    h_\mu(t) \coloneq  \frac{X_\mu}{\sqrt{2\pi}c} e^{-\frac{(t - (\mu \frac{P}{S}+\frac{P}{2S}))^2}{2c^2}}=\frac{\theta_\mu}{2\sqrt{2\pi}c\, \erf{\frac{P}{2S\sqrt{2}c}}} e^{-\frac{(t - (\mu \frac{P}{S}+\frac{P}{2S}))^2}{2c^2}}. 
\end{equation}

This implies that the error will be proportional to the exponential decaying tails of each Gaussian. Here $c \in \mathbb{R}$ is the standard deviation which we take to be the same for all the functions. The parameter $X_\mu$ is chosen to satisfy the normalization condition of the functions $h_\mu$. The normalisation factor of $\frac{1}{\sqrt{2\pi}c}$ is included for the ease of calculation later on. 
 Hence, we define the intermediate time-dependent Hamiltonian $H_{\rm diff}(t)$ as follows
\begin{equation} \label{eq:H_diff}
    H_{\rm diff}(t) \coloneq \sum_{\mu=0}^{S-1} h_\mu^{\rm ext}(t)\, O_\mu = \sum_{\mu=0}^{S-1}  \sum_{l \in \mathbb{Z}}\frac{\theta_\mu}{2\, \erf{\frac{P}{2S\sqrt{2}c}}} \frac{1}{\sqrt{2\pi}c} e^{-\frac{(t - (\mu \frac{P}{S}+\frac{P}{2S}) + lP)^2}{2c^2}} O_\mu
\end{equation}
 Notice that there is no longer the indicator functions $\mathbbm{1}_{[\frac{\mu P}{S}, \frac{(\mu + 1) P}{S}]}(t)$ in the newly defined Hamiltonian.

% As the analytical treatment is the same for functions $f_j^{BC}(t)$ and for $g^{BC}(t)$, a convenient way to rewrite this $H_{\rm BC}(t)$ is to rename the index to $\mu$ which indexes the operators in the Hamiltonians in the order in which the gates are applied in the circuit.

% As shown above in Section 2, $H_{\rm BC}(t)$ can be written as 
% \begin{equation} \tag{\ref{eq: H_BC_mu}}
%      H_{\rm BC}(t) =\sum_{\mu = 0}^{S-1} \mathbbm{1}_{[\frac{\mu P}{S}, \frac{(\mu + 1) P}{S}]}(t) X_\mu \frac{1}{\sqrt{2\pi}c} \e{-\frac{(t - (\mu \frac{P}{S}+\frac{P}{2S}))^2}{2c^2}} O_\mu  
% \end{equation}
% and similarly 
% \begin{equation}
%     H_{diff, per}(t) =\sum_{l\in \mathbb{Z}} \sum_{\mu = 0}^{S-1} X_\mu \frac{1}{\sqrt{2\pi}c} \e{-\frac{(t - (\mu \frac{P}{S}+\frac{P}{2S}) + lP)^2}{2c^2}} O_\mu  
% \end{equation}
% % and $H_{trun}(t)$ as \rs{Fill}

% To study how well $H_{trun}(t)$ approximates $H_{\rm BC}(t)$, we will explicitly bound the $L^1$ distance between these two Hamiltonians (see Section \ref{sec:integral_error_appendix}) by calculating their respective $L^1$ norms to the smooth periodic Hamiltonian $H_{diff, per}(t)$ defined above. 

\subsubsection{Bounded error between $H_{\rm BC}$ and $H_{\rm diff}$}
In this subsection, let us take a look at the error between the evolutions generated by $H_{\rm BC}(t)$ and $H_{\rm diff}(t)$, which is bounded by the difference between these Hamiltonians by virtue of \cref{eq:error}. Recall that the indicator functions $\mathbbm{1}_{[\frac{\mu P}{S}, \frac{(\mu + 1) P}{S}]}(t)$ present in the coefficient functions of $H_{\rm BC}(t)$ come from the fact that in a circuit, each operator is turned on only on specific time stamps of duration $\frac{P}{S}$ during computational time $P$. Hence, the difference between the superposition of Gaussian functions and their concatenation, which comes from the series of Gaussian pulses in the circuit, is simply the accumulative area of their tails.

With this in mind, we prove the following lemma:

\begin{lemma} \label{lem:area_error}
    Given a set of $S$ Gaussian functions $\{h_\mu(t)\}_\mu$ described in \cref{eq:hmu} with the same standard deviation $c \in \mathbb{R}$ and different peak heights $X_\mu$
    \begin{equation} \tag{\ref{eq:hmu}}
        h_\mu(t) \coloneq  \frac{X_\mu}{\sqrt{2\pi}c} e^{-\frac{(t - (\mu \frac{P}{S}+\frac{P}{2S}))^2}{2c^2}} 
    \end{equation}
    with equally spaced peaks on $S$ regular segments of the total interval $P$. Denote the periodic, differentiable extension of these functions as $\{h_\mu^{\rm ext}(t)\}_\mu$, where $h_\mu^{\rm ext}(t) = \sum_{l \in \mathbb{Z}} h_\mu(t+lP)$ for all $0\leq \mu \leq S-1$.
    
    The difference in area ($L^1$ normed difference) between the periodic, differentiable extension $\{h_\mu^{\rm ext}(t)\}_\mu$ and its concatenated counterpart $\{\mathbbm{1}_{[\frac{\mu P}{S}, \frac{(\mu +1)P}{S}]}h_\mu(t)\}_\mu$ on the interval $[0,P]$ is bounded
    \begin{equation} \label{ineq:area_error}
        \int_0^P \left |\sum_{\mu=0}^{S-1}h_\mu^{\rm ext}(t) - \sum_{\mu=0}^{S-1}\mathbbm{1}_{[\frac{\mu P}{S}, \frac{(\mu + 1) P}{S}]}(t)h_\mu(t)\right |\, dt \leq \max_\mu \theta_\mu \frac{SP}{2} \left(  \frac{1}{{\rm erf}{\frac{P}{2S\sqrt{2}c}}} -1 \right )
        % \frac{2\sqrt{2}}{\sqrt{\pi}}\max_\mu X_\mu  \, \frac{S^2c}{P } \e{-\frac{P^2}{(2\sqrt{2}Sc)^2}}
    \end{equation}
\end{lemma}

\begin{proof}
    It is easy to see from the definition that for all $0 \leq \mu \leq S-1$, we have$\ h_\mu^{\rm ext}(t) \geq \mathbbm{1}_{[\frac{\mu P}{S}, \frac{(\mu + 1) P}{S}]}(t)h_\mu(t)$ for $t \in [0,P]$. Hence, for all $0 \leq \mu \leq S-1$, the $L^1$ normed difference of $h_\mu^{\rm ext}(t)$ and $\mathbbm{1}_{[\frac{\mu P}{S}, \frac{(\mu + 1) P}{S}]}(t)h_\mu(t)$ is simply the difference between the two areas under these two curves. By construction, the area under the $h_\mu^{\rm ext}(t)$ is 
\begin{align}
    \int_0^P h_\mu^{\rm ext}(t) \, dt = \int_0^P \sum_{l \in \mathbb{Z}}h_\mu(t-lP) \, dt = \int_{-\infty}^{\infty} h_\mu(t) \, dt = \frac{\theta_\mu}{2\, \erf{\frac{P}{2S\sqrt{2}c}}}
\end{align}
The second equality is from a simple change of variable, whereas the last equality uses the explicit definition of $h_\mu(t)$ in \cref{eq:hmu} and the definition of $X_\mu$. On the other hand, the area under coefficient functions of $H_{\rm BC}(t)$, $\mathbbm{1}_{[\frac{\mu P}{S}, \frac{(\mu + 1) P}{S}]}(t)h_\mu(t)$ is $\frac{\theta_\mu}{2}$, required by \cref{lem:first}. Hence, for $0 \leq \mu \leq S-1$, the $L^1$ norm difference is $ \int_0^P \left |h_\mu^{\rm ext}(t) - \mathbbm{1}_{[\frac{\mu P}{S}, \frac{(\mu + 1) P}{S}]}(t)h_\mu(t)\right |\ dt = \frac{\theta_\mu}{2} \left(  \frac{1}{\erf{\frac{P}{2S\sqrt{2}c}}} -1 \right ) $.

Hence, the $L^1$ normed difference between the periodic, differentiable extension $\{h_\mu^{\rm ext}(t)\}_\mu$ and its concatenated counterpart $\{\mathbbm{1}_{[\frac{\mu P}{S}, \frac{(\mu +1)P}{S}]}h_\mu(t)\}_\mu$ on the interval $[0,P]$ is bounded by
\begin{align}
    \int_0^P \left |\sum_{\mu=0}^{S-1}h_\mu^{\rm ext}(t) - \sum_{\mu=0}^{S-1}\mathbbm{1}_{[\frac{\mu P}{S}, \frac{(\mu + 1) P}{S}]}(t)h_\mu(t)\right |\, dt &\leq  \sum_{\mu=0}^{S-1} \int_0^P \left |h_\mu^{\rm ext}(t) - \mathbbm{1}_{[\frac{\mu P}{S}, \frac{(\mu + 1) P}{S}]}(t)h_\mu(t)\right |\, dt \nonumber \\
    &\leq \max_\mu \theta_\mu \frac{S}{2} \left(  \frac{1}{\erf{\frac{P}{2S\sqrt{2}c}}} -1 \right )
    \end{align}

The first inequality is the triangle inequality, and the second inequality comes directly from the difference of the areas under the two curves calculated above.

\end{proof}

We can now bound the difference between $H_{\rm BC}(t)$ and $H_{\rm diff}(t)$. Using the definitions of the two Hamiltonians by \cref{eq:H2ways} and \cref{eq:H_diff}, we have
\begin{align}
    E_{\rm Area} \coloneq \int_0^P \|H_{\rm BC}(t) - H_{\rm diff}(t)\|dt = \int_0^P \left \|\sum_{\mu=0}^{S-1} \mathbbm{1}_{[\frac{\mu P}{S}, \frac{(\mu + 1) P}{S}]}(t)h_\mu(t) O_\mu  - \sum_{\mu=0}^{S-1} h_\mu^{\rm ext}(t) O_\mu \right \| \ dt 
\end{align}
By using Hölder's inequality, and taking the largest parameter $\max_\mu \|O_\mu\|$, the above equality can be bounded by 
\begin{align}
    E_{\rm Area}  &\leq \max_\mu\|O_\mu\|\int_0^P \left |\sum_{\mu=0}^{S-1}h_\mu(t) - \sum_{\mu=0}^{S-1}\mathbbm{1}_{[\frac{\mu P}{S}, \frac{(\mu + 1) P}{S}]}(t)h_\mu(t)\right |\, dt \leq \max_\mu\|O_\mu\| \max_\mu \theta_\mu \frac{SP}{2} \left(  \frac{1}{\erf{\frac{P}{2S\sqrt{2}c}}} -1 \right )
\end{align}
The last inequality is a direct result of \cref{lem:area_error}.

\subsubsection{Bounded error between $H_{\rm tr}^{[M]}$ and $H_{\rm diff}$}

% Recall the definition of $H_{\rm tr}^{[M]}$ and $H_{\rm diff}$
% $H_{\rm diff}(t) \coloneq \sum_{\mu=0}^{S-1}h_\mu(t) O_\mu$ 

The coefficient functions of $H_{\rm tr}^{[M]}(t)$ are just the Fourier truncations of those of $H_{\rm BC}(t)$. Similarly to how $H_{\rm BC}(t)$ can be written in two different ways as stated in \cref{eq:H2ways}, it is straightforward to rewrite $H_{\rm tr}^{[M]}(t)$ as follows
\begin{align}
    H_{\rm tr}^{[M]}(t) \coloneq \sum_{\mu = 0}^{S-1} h_\mu^{[M]}(t) O_\mu \quad \text{where} \quad  h_\mu^{[M]}(t)\coloneq \sum_{m=-M}^{M}e^{i\frac{2\pi }{P}mt}h_\mu^{(m)}
\end{align}
with $h_\mu^{(m)}$ being the Fourier coefficient of $h_\mu^{\rm ext}(t)$, the periodic and differentiable extension of $h_\mu(t)$. This is consistent with the identifications of the $f_{j}^{\rm BC}(t), g^{\rm BC}(t)$ with $h_\mu (t)$, and the definitions of $f_j^{(m)}, g^{(m)}$ in \cref{eq:F_coeff}. 

By \cref{eq:error}, to quantify the error between the time evolutions generated by  $H_{\rm tr}^{[M]}$ and $H_{\rm diff}$, we can bound it by the quantity
\begin{align} \label{eq:Hdiff,tr}
    E_{\rm Fourier} \coloneq  \int_0^P \| H_{\rm diff}(t)- H_{\rm tr}^{[M]}(t) \| dt, \quad \text{which, by definition, is  } \quad \int_0^P \left \| \sum_{\mu=0}^{S-1}h_\mu^{\rm ext}(t) O_\mu - \sum_{\mu = 0}^{S-1} h_\mu^{[M]}(t) O_\mu \right \| dt 
\end{align}
Essentially, the error difference of these two generators is caused by the truncation of the Fourier series of the coefficient functions of $H_{\rm diff}(t)$, which we have taken to be Gaussian functions. This truncation error can be made arbitrarily small, precisely because these coefficient functions are differentiable. This is the motivation to define the differentiable Hamiltonian $H_{\rm diff}(t)$ in the first place. Using a known result from complex analysis, we know  that the Fourier partial sum of a periodic analytic function converges exponentially fast in the truncation. To make this work self-contained, we state and prove the result applied to the specific case of Gaussian functions, defined in \cref{eq:hmu}, in the following lemma:
% \begin{lemma}
%     Given a Gaussian function $h_\mu(t)$ described in \cref{eq:hmu} with the same standard deviation $c \in \mathbb{R}$ and different peak heights $X_\mu$. Its periodic extensions is  
%     \begin{align}
%     h_\mu^{\rm ext}(t) &\coloneq \sum_{l\in \mathbb{Z}}h_\mu(t + lP) = \sum_{l\in \mathbb{Z}} X_\mu \frac{1}{\sqrt{2\pi}c} \e{-\frac{(t - (\mu \frac{P}{S}+\frac{P}{2S})+lP)^2}{2c^2}}
%     \end{align}
%     The Fourier partial sum of $h_\mu^{\rm ext}(t)$ converges to $h_\mu^{\rm ext}(t)$ exponentially fast.
% % \begin{equation}
% %    |h_\mu^{per}(t) - h_\mu^{trun}(t)| = \left|\sum_{n=M+1}^\infty\hat{q}_n\e{i\frac{2\pi}{P}nt}\right| \leq \sum_{n=M+1}^\infty  \frac{X_\mu}{P} \e{\frac{a^2}{2c^2}} \e{-2\pi \frac{n}{P}a} \leq  \frac{X_\mu}{P}\e{\frac{a^2}{2c^2}} \int_{M}^\infty \e{-2\pi \frac{s}{P}a} ds = \e{\frac{a^2}{2c^2}} \frac{X_\mu}{2\pi a} \e{-2\pi \frac{M}{P}a} 
% % \end{equation}

% \end{lemma}

\begin{lemma} \label{lem:F_error}
    Given a set of $S$ Gaussian functions $\{h_\mu(t)\}_\mu$ described in \cref{eq:hmu} with the same standard deviation $c \in \mathbb{R}$ and different peak heights $X_\mu$
    \begin{equation} \tag{\ref{eq:hmu}}
        h_\mu(t) \coloneq  \frac{X_\mu}{\sqrt{2\pi}c} e^{-\frac{(t - (\mu \frac{P}{S}+\frac{P}{2S}))^2}{2c^2}} 
    \end{equation}
    with equally spaced peaks on $S$ regular segments of the total interval $P$. Denote the periodic, differentiable extension of these functions as $\{h_\mu^{\rm ext}(t)\}_\mu$, where $h_\mu^{\rm ext}(t) = \sum_{l \in \mathbb{Z}} h_\mu(t+lP)$ for all $0\leq \mu \leq S-1$.

    The Fourier partial sum of $h_\mu^{\rm ext}(t)$ converges to $h_\mu^{\rm ext}(t)$ exponentially fast.
\end{lemma}

\begin{proof}
    % Given the set of $S$ Gaussian functions $\{h_\mu(t)\}_\mu$ described in \cref{eq:Xmu}, we will show that, for each $0\leq \mu \leq S-1$, the upper bound of the difference between $h_\mu^{\rm ext}(t)$ and its Fourier partial sum is independent of $\mu$. Hence, for the ease of notation, we forgo the subscript $\mu$ when discussing the periodic extensions of these Gaussian functions $h_\mu(t)$. 
    
    By definition, the Fourier coefficient of $h_\mu^{\rm ext}(t)$ is 
\begin{align}
    h_\mu^{(m)} = \frac{1}{P} \int_0^P h_\mu^{\rm ext}(t) \e{-i2\pi \frac{m}{P}t} dt 
\end{align}

% \begin{align}
%     \hat{q}_n = \frac{1}{P} \int_0^P h_\mu^{\rm ext}(t) \e{-i2\pi \frac{n}{P}t} dt 
% \end{align}

% $= \frac{1}{P} \int_0^P \sum_{l\in \mathbb{Z}} X_\mu \frac{1}{\sqrt{2\pi}c} \e{-\frac{(t - (\mu \frac{P}{S}+\frac{P}{2S})+lP)^2}{2c^2}} \e{-i2\pi \frac{n}{P}t} dt$

To prove that the Fourier coefficient $h_\mu^{(m)}$ decay exponentially, we go to the complex plane, where $h_\mu^{(m)}$ is a term in the integral along a rectangle. Note that $h_\mu^{\rm ext}(z)\e{-i2\pi \frac{n}{P}z}$ is holomorphic in the complex plane $\mathbb{C}$. By Cauchy Integral Theorem, for $a \in \mathbb{R}$ (quantifying how much the contour goes into the complex plane), the integral of this function over a closed contour is 0: 
\begin{align}
    &\oint h_\mu^{\rm ext}(z)\e{-i2\pi \frac{n}{P}z} dz = \nonumber \\
    &\int_0^P h_\mu^{\rm ext}(z)\e{-i2\pi \frac{n}{P}z} dz + \int_P^{P-ia} h_\mu^{\rm ext}(z)\e{-i2\pi \frac{n}{P}z} dz + \int_{P-ia} ^ {-ia} h_\mu^{\rm ext}(z)\e{-i2\pi \frac{n}{P}z} dz + \int_{-ia}^0 h_\mu^{\rm ext}(z)\e{-i2\pi \frac{n}{P}z} dz = 0
\end{align}
By the periodic construction of $h_\mu^{\rm ext}(z)$, the second and fourth terms cancel. To study the first term, which is the Fourier coefficient of interest (scaled by a factor of $P$), we evaluate the third term. We do this by direct substitution of $h_\mu^{\rm ext}(z)$:
\begin{align}
    \int_{P-ia} ^ {-ia} h_\mu^{\rm ext}(z)\e{-i2\pi \frac{n}{P}z} dz = X_\mu \frac{1}{\sqrt{2\pi}c} \sum_{l\in \mathbb{Z}} \int_{P-ia}^{-ia} \e{-\frac{(z - (\mu \frac{P}{S}+\frac{P}{2S})+lP)^2}{2c^2}} \e{-i2\pi \frac{n}{P}z} \, dz 
\end{align}
Since the integral is over a horizontal line (with the imaginary part kept constant), we use the change of variable $z = x - ia$ to change this into an integral over a real variable. The sum of the integrals is
\begin{align}
    \sum_{l\in \mathbb{Z}} \int_P^0 \e{-\frac{(x-ia - (\mu \frac{P}{S}+\frac{P}{2S})+lP)^2}{2c^2}} \e{-i2\pi \frac{n}{P}(x-ia)} \, dx = \sum_{l\in \mathbb{Z}} \int_{P-(\mu \frac{P}{S}+\frac{P}{2S}) + lP}^{-(\mu \frac{P}{S}+\frac{P}{2S}) + lP} \e{-\frac{s^2}{2c^2}} \e{\frac{a^2}{2c^2}} \e{\frac{ias}{c^2}} \e{-i2\pi \frac{n}{P}(s + (\mu \frac{P}{S}+\frac{P}{2S}) - lP)} \e{-2\pi \frac{n}{P}a} \, ds
\end{align}
The equality comes from expanding the square to separate the product into the real and complex exponentials and letting the real part of the exponent to be $s$. Once we have separated the real and complex exponentials, it is easy to see that all complex exponential become $1$ when taking the modulus, and the only part left to consider are the real exponentials. Overall, by using the triangle inequality, the modulus of the third term is bounded by
\begin{align}
    \left|\int_{P-ia} ^ {-ia} h_\mu^{\rm ext}(z)\e{-i2\pi \frac{n}{P}z} dz \right| &\leq X_\mu \frac{1}{\sqrt{2\pi}c} \e{\frac{a^2}{2c^2}}  \e{-2\pi \frac{n}{P}a} \sum_{l\in \mathbb{Z}} \int^{P-(\mu \frac{P}{S}+\frac{P}{2S}) + lP}_{-(\mu \frac{P}{S}+\frac{P}{2S}) + lP} \e{-\frac{s^2}{2c^2}}  \, ds  
\end{align}
Note that the integrals in the sum in the last expression have supports which combine to be the whole real line. Hence, the sum of integrals over segments $[-(\mu \frac{P}{S}+\frac{P}{2S}) + lP, P-(\mu \frac{P}{S}+\frac{P}{2S}) + lP]$ is equivalent to one integral of the same function, over the real line. Hence, the upper bound of the third term is
\begin{align}
    \left|\int_{P-ia} ^ {-ia} h_\mu^{\rm ext}(z)\e{-i2\pi \frac{n}{P}z} dz \right| \leq X_\mu \frac{1}{\sqrt{2\pi}c}  \e{\frac{a^2}{2c^2}}  \e{-2\pi \frac{n}{P}a}  \int_{-\infty}^\infty  \e{-\frac{s^2}{2c^2}}  \, ds  
    = X_\mu \e{\frac{a^2}{2c^2}}  \e{-2\pi \frac{n}{P}a} 
\end{align}
The last equality comes from the normalisation of the standard Gaussian function. As mentioned above, the first and last terms of the contour integral add up to 0. Hence, the first term, the Fourier coefficient of interest, and the third term share the same upper bound. 
\begin{align}
    \Rightarrow |h_\mu^{(m)}| &= \frac{1}{P} \left|\int_{P-ia} ^ {-ia} h_\mu^{\rm ext}(z)\e{-i2\pi \frac{n}{P}z} dz \right|\leq \frac{X_\mu}{P} \e{\frac{a^2}{2c^2}}  \e{-2\pi \frac{n}{P}a}
\end{align}

We have bounded the Fourier coefficient $h_\mu^{(m)}$. The last step of the proof is to see how the Fourier partial sum converges. Denote $h_\mu^{[M]}(t)$ the $M$-th partial sum of the Fourier Series of $h_\mu^{\rm ext}(t)$, i.e. $h_\mu^{[M]}(t) = \sum_{m=-
M}^M h_\mu^{(m)} \e{i\frac{2\pi}{P}mt}$. We have
\begin{equation}
   |h_\mu^{\rm ext}(t) - h_\mu^{[M]}(t)| \leq 2\left|\sum_{n=M+1}^\infty h_\mu^{(m)}\e{i\frac{2\pi}{P}nt}\right| \leq 2\sum_{n=M+1}^\infty  \frac{X_\mu}{P} \e{\frac{a^2}{2c^2}} \e{-2\pi \frac{n}{P}a} \leq  2\frac{X_\mu}{P}\e{\frac{a^2}{2c^2}} \int_{M}^\infty \e{-2\pi \frac{s}{P}a} ds = \e{\frac{a^2}{2c^2}} \frac{X_\mu}{\pi a} \e{-2\pi \frac{M}{P}a} 
\end{equation}
The first and second inequalities are the triangle inequality, the third inequality is one between sums and integrals, the last equality is the result of direct computation of the integral. We have proven that the Fourier partial sum $h_\mu^{[M]}(t)$ converges $h_\mu^{\rm ext}(t)$ exponentially fast.

\end{proof}

We now bound the the error between the time evolutions generated by  $H_{\rm tr}^{[M]}$ and $H_{\rm diff}$
% \begin{align} 
%     E_{\rm Fourier} \coloneq  \int_0^P \| H_{\rm diff}(t)- H_{\rm tr}^{[M]}(t) \| dt = \int_0^P \left \| \sum_  {\mu=0}^{S-1}h_\mu^{\rm ext}(t) O_\mu - \sum_{\mu = 0}^{S-1} h_\mu^{[M]}(t) O_\mu \right \| dt 
% \end{align}
By using Hölder's inequality, and taking the largest parameter $\max_\mu \|O_\mu\|$, the above equality can be bounded by 
\begin{align}\label{eq:Fourier_trunc}
    E_{\rm Fourier} \leq \max_\mu\|O_\mu\|\sum_{\mu = 0}^{S-1}\int_0^P \left |h_\mu^{\rm ext}(t) - h_\mu^{[M]}(t)\right |\, dt \leq \max_\mu\|O_\mu\| \max_\mu \theta_\mu \frac{SP}{2\pi a} \frac{1}{ \erf{\frac{P}{2S\sqrt{2}c}}} \e{\frac{a^2}{2c^2}} \e{-2\pi \frac{M}{P}a} 
\end{align}
The last inequality is a direct result of \cref{lem:F_error}.

We have found the explicit expression for the total error occurred when trying to simulate the effect of a quantum circuit $\mathcal{C}$, which can be written as a time-dependent Hamiltonian $H_{\rm {BC}}(t)$, with the time evolution of time-independent Hamiltonian $H_{\rm {indep}}$
\begin{align} \tag{\ref{ineq:error}}
    E&\coloneq\| \mathcal{T}e^{-i\int_0^PH_{\rm BC}(s)ds}-\mathcal{T}e^{-i\int_0^PH_{\rm tr}^{[M]}(s)ds}\| \leq \int_0^P\left(\| H_{\rm BC}(s)-H_{\rm diff}(s)\|+\| H_{\rm diff}(s)-H_{\rm tr}^{[M]}(s)\|\right)ds \\
    &\leq \max_\mu \|O_\mu\|\max_\mu \theta_\mu \frac{S}{2} \frac{1}{\V}\left [ 1 - \V \right ] + \max_\mu\|O_\mu\| \max_\mu \theta_\mu \frac{SP}{2\pi a} \frac{1}{ \erf{\frac{P}{2S\sqrt{2}c}}} \e{\frac{a^2}{2c^2}} \e{-2\pi \frac{M}{P}a} \\
    &= \max_\mu \|O_\mu\|\max_\mu \theta_\mu \frac{S}{2} \left ( \frac{1}{\V}\left [ 1 - \V \right ] + \frac{P}{\pi a} \frac{1}{\V} \e{\frac{a^2}{2c^2}} \e{-2\pi \frac{M}{P}a}\right ) \label{eq:total_error}
    % &= \frac{S\theta}{2} \frac{1}{\V} \left[ 1 - \V + \frac{P}{2\pi a}  e^{\frac{a^2}{2c^2}} e^{- 2\pi \frac{M}{P}a} \right ]
\end{align}
As the quantities $\max_\mu \|O_\mu\|, \, \max_\mu \theta_\mu$ and $S$ are fixed according to the quantum circuit $\mathcal{C}$ that we want to carry out, the error now depends only on the expression in the brackets. We denote $\max_\mu \|O_\mu\|, \, \max_\mu \theta_\mu$ as $Y, \theta$ respectively.

We now prove that the error can be made arbitrarily small by restricting the ratio of $P$, the total computational time, and $c$, the standard deviation of the Gaussian functions, to scale with $S$. Intuitively, smaller $c$ means less overlapping between Gaussian functions, and larger $P$ means larger interval on which they overlap. The quantity $a$ should not be increasing with $S$ due to the exponential present in the error between $H_{\rm diff}(t)$ and $H^{[M]}_{\rm tr}(t)$.

With this intuition in mind, we prove the following lemma:

% given an arbitrarily small error threshold $\varepsilon$, we come to the following ansatz:
\begin{lemma}
    
Given a quantity of error $E$ defined in \cref{eq:total_error}:
\begin{equation} \tag{\ref{eq:total_error}}
    % E = \max_\mu \|O_\mu\|\max_\mu \theta_\mu \frac{S}{2} \left ( \frac{1}{\V}\left [ 1 - \V \right ] + \frac{P}{\pi a} \frac{1}{\V} \e{\frac{a^2}{2c^2}} \e{-2\pi \frac{M}{P}a}\right ) 
    E = Y \, \theta\, \frac{S}{2} \left ( \frac{1}{\V}\left [ 1 - \V \right ] + \frac{P}{\pi a} \frac{1}{\V} \e{\frac{a^2}{2c^2}} \e{-2\pi \frac{M}{P}a}\right ) 
\end{equation}

For any threshold $\varepsilon >0$, by taking the ratio $\frac{P}{S}$ constant, and $\frac{P}{c}$ to scale with $S$ as follows
\begin{align} \label{eq:ansatz}
\begin{cases}
    a \coloneq  \frac{P}{2S}\quad\mbox{constant} \\
    \frac{P}{c} \coloneq  S^{\alpha + 1} \text{ with } \alpha = \frac{1}{\ln{S}} \left ( \ln{2\sqrt{2} + \ln{\sqrt{\ln{\left( \frac{Y\theta S}{\varepsilon} + 1 \right)}}}} \right )
\end{cases}
\end{align}
% \rs{add statement of theorem, reorganize this part}
and with $M = \mathcal{O(S \ln{S})}$, the total error $E$ will be below the threshold $\varepsilon$ i.e. $E \leq \varepsilon$.

% \begin{align} \label{eq:M}
%     M \geq \frac{S}{\pi} \left [ 2\ln{\left( \frac{S\theta}{\varepsilon} + 1 \right)} + \ln{S} + \ln{\frac{2}{\pi}} \right]
% \end{align}

\end{lemma}

\begin{proof}
We will do this by proving that each of the two error contributions will be at most $\frac{\varepsilon}{2}$.
% \begin{align}
%     \frac{P}{2\sqrt{2}Sc} &= \frac{S^\alpha}{2\sqrt{2}} = \frac{S^{\left ( \frac{1}{\ln{S}} \left ( \ln{2\sqrt{2} + \ln{\sqrt{\ln{\left( \frac{Y\theta S}{\varepsilon} + 1 \right)}}}} \right ) \right )}}{2\sqrt{2}} = \sqrt{\ln{\left( \frac{Y\theta S}{\varepsilon} + 1 \right)}}  
% \end{align}

First, using the new ansatz in \cref{eq:ansatz}, we rewrite 
% $\frac{P}{2\sqrt{2}Sc}$ as $\frac{S^\alpha}{2\sqrt{2}}$. 
\begin{equation} \label{eq:alpha1}
    \frac{P}{2\sqrt{2}Sc} = \frac{S^\alpha}{2\sqrt{2}}
\end{equation}
This comes from direct substitution of the definition of $\alpha$ in \cref{eq:ansatz}. Rewriting the error between time evolutions of $H_{\rm BC}(t)$ and of $H_{\rm diff}(t)$ with the new ansatz, we have
\begin{align}
    E_{\rm{Area}} &= \frac{Y \theta S}{2} \frac{1}{\erf{\frac{S^\alpha}{2\sqrt{2}}}}\left [ 1-\erf{\frac{S^\alpha}{2\sqrt{2}}} \right ]
\end{align}

As mentioned above, we want to bound this error $E_{\rm Area}$ by $\frac{\varepsilon}{2}$. Using a known inequality of the error function, where
\begin{equation} \label{ineq:erf}
    1 - \erf{x} \leq \e{-x^2} \quad {\rm for} \quad x > 0
\end{equation} 
and applying that to $\frac{P}{2\sqrt{2}Sc}$, it is easy to see that $E_{\rm Area}$ is upper bounded by $\frac{Y \theta S}{2} \frac{e^{-\frac{S^{2\alpha}}{8}}}{1 - e^{-\frac{S^{2\alpha}}{8}}}\ $. Hence, to make $E_{\rm Area}$ less than $\frac{\varepsilon}{2}$, it is sufficient make the upper bound of $E_{\rm Area}$ less than $\frac{\varepsilon}{2}$. This is done by taking $\alpha$ as in \cref{eq:ansatz} and subsequently substituting  \cref{eq:alpha1} into the upper bound of $E_{\rm Area}$.

% With the definition in \cref{eq:ansatz}, by rearranging, we see that $\alpha$ satisfies the following inequality \begin{align}
%     \frac{S \theta}{2} \frac{e^{-\frac{S^{2\alpha}}{8}}}{1 - e^{-\frac{S^{2\alpha}}{8}}}\ \leq \frac{\varepsilon}{2}
% \end{align} 

% $E_{\rm Area} \leq \frac{\varepsilon}{2}$.
% We now prove that $\alpha$ satisfies a stricter constraint, whihc implies that $E_{\rm{Area}}$ being arbitrarily small. This strict constraint is 
% \begin{align}
%     \frac{S \theta}{2} \frac{e^{-\frac{S^{2\alpha}}{8}}}{1 - e^{-\frac{S^{2\alpha}}{8}}}\ \leq \frac{\varepsilon}{2} \quad \text{i.e.} \quad \frac{S \theta}{2} \frac{e^{-\left (\sqrt{\ln{\left( \frac{S\theta}{\varepsilon} - 1 \right)}}\right )^2}}{1 - e^{-\left (\sqrt{\ln{\left( \frac{S\theta}{\varepsilon} - 1 \right)}}\right )^2}}\ \leq \frac{\varepsilon}{2}
% \end{align}
% which is satisfied by the definition of $\alpha$, easily seen by the direct substitution into Equation 9. Using the inequality \textcolor{red}{number}, where for $x \geq 0, \ \erf{x} \leq 1-e^{-x^2}$
% \textcolor{red}{Change the order to make it sound smoother}, and use $\frac{P}2\sqrt{2}{Sc}$ as the argument, we have shown that $E_{\rm{Area}} \leq \frac{\varepsilon}{2}$.
% \textcolor{red}{See which lines you want to include, re phrase a little bit}

Rewriting the Fourier error $E_{\rm{Fourier}}$ with the new ansatz in \cref{eq:ansatz}, we have 
\begin{align}
    E_{\rm{Fourier}} &= \frac{Y  \theta S^2}{\pi} \frac{1}{\erf{\frac{S^{\alpha}}{2\sqrt{2}}}} e^{\frac{S^{2\alpha}}{8}} e^{-\pi \frac{M}{S}} 
\end{align}
Again, using the inequality for the erf function in \cref{ineq:erf} and taking $\frac{S^{2\alpha}}{8}$ to be the argument, we can bound the Fourier error $E_{\rm{Fourier}}$ by 
\begin{align}
    E_{\rm{Fourier}} &\leq \frac{Y \theta S^2}{\pi} \frac{1}{1 - e^{-\frac{S^{2\alpha}}{8}}} e^{\frac{S^{2\alpha}}{8}} e^{-\pi \frac{M}{S}} = \frac{S \varepsilon}{\pi}  \left ( \frac{Y\theta S}{\varepsilon}+1 \right )^2 e^{-\pi \frac{M}{S}}
\end{align}
The last equality comes from a direct substitution of the expression for $\alpha$ in \cref{eq:ansatz} into the terms $\frac{S^{2\alpha}}{8}$, with which we have $ \frac{S^{2\alpha}}{8} = \ln{\left( \frac{S\theta}{\varepsilon} + 1 \right)} $. We can make the upper bound of $E_{\rm Fourier}$ less than $\frac{\varepsilon}{2}$ by taking $M$ to be 
\begin{align} \label{eq:M}
    M \geq \frac{S}{\pi} \left [ 2\ln{\left( \frac{Y\theta S}{\varepsilon} + 1 \right)} + \ln{S} + \ln{\frac{2}{\pi}} \right]
\end{align}
directly substituting this into the inequality above, we can see that $E_{\rm Fourier}$ is indeed less than $\frac{\varepsilon}{2}$ as required.

% We conclude that with the ansatz given, where 
% \begin{align} \tag{\ref{eq:ansatz}}
%     a \coloneq \frac{P}{2S}\ ; \quad \frac{P}{c} \coloneq S^{\alpha + 1} \text{ with } \alpha = \frac{1}{\ln{S}} \left ( \ln{2\sqrt{2} + \ln{\sqrt{\ln{\left( \frac{S\theta}{\varepsilon} + 1 \right)}}}} \right )
% \end{align}
% and with the number of modes scaling as $\mathcal{O}(S\ln(S))$, the error can be made arbitrarily small.

\end{proof}

% This result shows that the time independent evolution of \cref{eq:H}, on a subspace can approximate to arbitrary accuracy the time dependent evolution which is universal.

This shows that the evolution of the impurity model of \cref{eq:H}, moving through the field of \cref{eq:B_field}, approximates to arbitrary accuracy, on the $k=0$ subspace and in the rest frame of the lattice, the evolution generated by the time-dependent Hamiltonian $H_{\rm BC}(t)$, which is universal.

\section{Discussion}\label{sec:conclusion}
We have shown that an impurity model with fixed couplings, moving with constant velocity through a static magnetic field that increases linearly in space, can perform universal quantum computation, with the size of the impurity scaling as $\mathcal{O}(S\ln{S})$ in the depth $S$ of the computation. Our proof builds upon the result of \cite{brod_childs2013computational}, where it is shown that matchgates on a graph that is neither a line or a cycle can perform universal computation. Mapping that result to a fermionic system implies that an impurity model under a time-dependent Hamiltonian is universal. Through a series of reductions we show that the required time dependence is generated by the Peierls phases that the hopping terms acquire as the lattice moves through the field, so that the resulting evolution approximates the time-dependent one with arbitrary accuracy.

Our technique did not reach an impurity of constant size: the factor $\frac{M}{P}$ in the error arising from the Fourier truncation (see \cref{eq:Fourier_trunc}) forces the impurity size to grow at least linearly with the circuit depth,  hence this is not a mere problem of optimisation of different parameters presented in this paper, such as the computational time $P$ or the standard deviation $c$ of Gaussian pulses. A constant impurity size therefore requires a different approach. Such an approach was recently given by Arunachalam \emph{et al.}~\cite{arunachalam2026impurity}, who showed that time evolution under a time-independent impurity Hamiltonian with a fixed impurity of $24$ fermionic modes is universal. In their construction the quartic interaction is the same for every computation and the circuit is encoded in the initial state, a Slater determinant of wave packets that propagate freely through the bath and undergo the same fixed unitary each time they pass the impurity. In our construction the roles are reversed: the circuit is encoded in the couplings of the Hamiltonian, through the Fourier coefficients $\fmj$ and $\gm$, while the initial state, a product state in the $k=0$ momentum modes, does not depend on the circuit.

The size $M$ of the impurity, with a bath of size $N$, interpolates between a non-interacting system ($M=0$), which can be simulated with $\mathcal{O}(N^3)$ classical resources, and a fully interacting system ($M=\mathcal{O}(N)$) that can perform universal quantum computation \cite{Childs_2013_QuantumWalk}. The result of \cite{arunachalam2026impurity} shows that universality already sets in at constant $M$, while in the circuit model systems with $\mathcal{O}(1)$ non-Gaussian gates are classically simulable in polynomial time in the size of the bath \cite{Dias2024classicalsimulation, dias2026}. 
Together these results suggest that the relevant resource is not the size of the impurity alone, but the number of interacting operations accumulated over the evolution. The same work \cite{arunachalam2026impurity} also gives polynomial-time classical algorithms for the ground-state and thermal properties of impurity models with constant impurity size, so the computational power of these models lies in their dynamics rather than in their equilibrium properties.

In summary, our results give a construction of universal impurity dynamics complementary to that of \cite{arunachalam2026impurity}, in which the computation is programmed through the couplings of the Hamiltonian rather than the initial state. It would be interesting to know whether a constant-size impurity also suffices when the computation is encoded in the Hamiltonian and the initial state is independent of the circuit, or whether an impurity that grows with the depth of the computation is necessary in that setting.

\section*{Acknowledgements and use of AI}

We thank Andrew M. Childs for looking through a draft of this manuscript and for helpful discussions.

The AI model Claude Opus 5.5 (Anthropic) was used to identify an error in the proof of \cref{thm:1} in the first version of this work: the time-independent Hamiltonian of \cref{subsec:construction_H_indep} did not reproduce the required time-dependent Hamiltonian for hopping terms that cross the periodic boundary of the rings. The corrected construction in \cref{subsec:construction_H_indep}, in which the lattice moves through a linearly increasing magnetic field, was conceived and calculated by the authors; Claude was used to verify the calculation and to help write up the section, including \cref{fig:field}. The remaining parts of the paper, which are essentially unchanged from the first version, did not involve the use of AI.

\printbibliography

\end{document}

%% file: admin/preamble2.tex
\usepackage[a4paper, total={7in, 9in}]{geometry}
\usepackage[utf8]{inputenc}
\usepackage{tikz}
\usepackage{background}
\usepackage{caption}
\usepackage{graphicx}
\usepackage{subfig}
\usepackage[normalem]{ulem}
\usepackage{comment}
\usepackage[linewidth=1pt]{mdframed}
\usepackage{longtable}

\backgroundsetup{
   scale=1,
   angle=0,
   opacity=1,
   color=black,
   contents={\begin{tikzpicture}[remember picture, overlay, font=\sffamily]
     \end{tikzpicture}}
 }

\usepackage[
    backend=bibtex,
    style=numeric,
    sorting=none
]{biblatex}
\bibliography{cite}

\usepackage{amsmath}
\usepackage{braket}

\usepackage{listings}
\usepackage{mathtools}
\usepackage[ruled,vlined]{algorithm2e}
\usepackage{amssymb}
\usepackage[ampersand]{easylist}
\usepackage{enumitem}
\usepackage{hyperref}
\hypersetup{
  colorlinks = true, %Colours links instead of ugly boxes
  linkcolor={blue!50!black}, %Colour of internal links
  citecolor=red, %{red!50!black}, %Colour of citations
  urlcolor=blue, %Colour for external hyperlinks
  linktoc=page,
}
\usepackage[capitalise, nameinlink]{cleveref}
\usepackage{amsthm}
\usepackage[
  style=long,
  nolist, 
  nonumberlist,
  nopostdot,
  acronym,
  toc=true, 
]{glossaries}
\usepackage{bm}
\usepackage{booktabs}
\usepackage{multirow}

\usepackage{array}
\usepackage[british]{babel}

\newtheorem{theorem}{Theorem}

\newtheorem{lemma}[theorem]{Lemma}

%% file: admin/macros.tex
\newcommand{\e}[1]{e^{#1}}

\newcommand{\jindex}{\sum_{j=1}^{N-2}}
\newcommand{\fmj}{f_j^{(m)}}
\newcommand{\gm}{g^{(m)}}

\newcommand{\cj}[2]{c_{j, #1}^\dag c_{j+1, #1 + #2}}
\newcommand{\cjdag}[2]{c_{j+1, #1 + #2}^\dag c_{j, #1}}
\newcommand{\cN}[2]{c_{N-2, #1}^\dag c_{N, #1 + #2}}
\newcommand{\cNdag}[2]{c_{N, #1 + #2}^\dag c_{N-2, #1}}
\newcommand{\cNN}[2]{c_{N-1, #1}^\dag c_{N-1, #2}}

\newcommand{\erf}[1]{\text{erf}\left( #1 \right)}

\newcommand{\V}{\text{erf}\left( \frac{P}{2\sqrt{2}Sc} \right)}